\documentclass[12pt,a4paper,normalheadings,headsepline,headinclude,bibtotoc,fleqn,review]{article}
\usepackage[latin1]{inputenc}
\usepackage{amsmath}
\usepackage{booktabs}
\usepackage{array}
\usepackage{amsthm}
\usepackage{dcolumn}
\usepackage{endnotes}
\usepackage{units}
\usepackage{marvosym}
\usepackage[lines=10]{geometry}
\usepackage{longtable}
\usepackage{rotating}
\usepackage{expdlist}
\usepackage{wrapfig}
\usepackage{expdlist}
\usepackage{amssymb}
\usepackage{gloss}
\usepackage{nomencl}
\usepackage{natbib}
\newtheorem{rem}{Remark}
\newtheorem{lemma}{Lemma}
\newtheorem{prop}{Proposition}

    \makegloss

\renewcommand{\thesection}{\arabic{section}}

\usepackage{fancyhdr} 
\begin{document}


\begin{center}\huge{Thought Viruses and Asset Prices}\footnote{I thank Dominik Grafenhofer and Carl Christian von Weizsäcker for helpful and encouraging discussions. This version: 30.12.2018.}\end{center} 
\begin{center} \emph{Wolfgang Kuhle}\\\emph{University of Economics, Prague, Czech Republic, E-mail: wkuhle@gmx.de\\ Max Planck Institute for Social Law and Social Policy, Munich, Germany.}
\end{center}




\noindent\emph{\textbf{Abstract:} We use insights from
epidemiology, namely the SIR model, to study how agents infect
each other with ``investment ideas." Once an investment idea
``goes viral," equilibrium prices exhibit the typical ``fever
peak," which is characteristic for speculative excesses as
described, e.g., in \citet{Kin00}. Using our model, we identify a
time line of symptoms that indicate whether a boom is in its early
or later stages. Regarding the market's top, we find that prices
start to decline while the number of infected agents, who buy the
asset, is still rising. Moreover, the presence of fully rational
agents (i) accelerates booms (ii) lowers peak prices and (iii) produces broad, drawn-out, market tops.}

\hspace{0.3cm}

\emph{``More and more people realized the misconception on which
the boom rested even as they continued to play the game."}
\citet{Sor94}, p. 57.

\section{Introduction}\label{sg1}

Market practitioners frequently argue that they observe asset
prices, which are ``far away" from
``fundamentals."\footnote{\citet{Sor94}, pp. 27-141, provides a
detailed argument and a large collection of case studies. See also
the ``Mr. Market" parable in \citet{Gra73}, pp. 188-213, or
\citet{Fis03}, pp. 266-275. \citet{Kin00} presents a historical
account of such boom and bust episodes.} They also suggest that
many such deviations come in the form of boom-bust cycles, which
are hard to reconcile with traditional asset pricing models. The
current paper presents a simple epidemiological model where
investors ``infect" each other with their investment ideas. The
model generates market prices that are in line with the
practitioners' boom and bust observation. In turn, we use the
model to perform comparative statics, and to establish a time-line
of events that lead and lag market tops.

To model the infectious spread of investment ideas, we borrow the
standard SIR model from the literature on
epidemiology.\footnote{\citet{Ker27} present the mainstream SIR
specification used here. \citet{Die02} review the historical
literature on modelling epidemics dating back to Daniel
Bernoulli's original manuscript, which was first circulated in
1760.} The SIR framework has been used successfully to model the
spread of infectious diseases, computer viruses, chain letters,
and religions. In the current setting it governs the rate with
which agents adopt and discard their views on certain assets. The
SIR model's key prediction, namely that the mass of infected
agents is hump-shaped, generates the boom-bust sequence in stock
prices that practitioners emphasize.


We distinguish two main scenarios. First, we study a setting where
susceptible investors buy the asset once infected, and sell once
cured. Second, we study a setting where cured agents form rational
expectations, which can induce them to hold overpriced assets in
order to wait for further appreciation. Such rational expectations
on the part of cured agents make prices (i) rise faster (ii) peak
earlier and (iii) generate broad, drawn-out, market tops. Without
rational agents, prices rise at a slower pace, but they reach a
higher maximum. Finally, we find that price peaks lead the peak in
infected agents across both scenarios. Put differently, prices go
into decline while sentiment is still improving.


Having derived these properties, we compare the observable parts
of our predictions to a few prominent cases of boom and bust.
Using search queries from Google as a proxy for the mass of
infected agents, we find that both, search queries and the
corresponding asset's price, show the fever curve pattern that the
current model suggests. Moreover, prices indeed tend to peak
before the mass of infected agents does. Finally, slow moving
boom-bust cycles, where cured agents have time to form rational
expectations, have broad tops.

\emph{Literature:} \citet{Daw95} suggests that humans are
susceptible to ``Viruses of the mind." Such viruses can come in
the form of religions or political believes, which spread in a
contagious manner.\footnote{In different contexts, \citet{Har95}
and \citet{Kah12} discuss how agents adopt and discard different,
not necessarily correct, approaches to decision making. See also
\citet{Wei71b}, who proposes a model of changing tastes.}

Instead of adopting religions, our agents believe in certain
investment ideas, valuation techniques, or future technological
developments. Moreover, our agents work as financial missionaries
when they share their investment tips with colleagues, friends,
and neighbors. The analogy to an infectious process is even more
straight forward in specific episodes such as the ``gold fever"
epidemic, where a susceptible east-coast population was infected
by the first prospectors who returned from California.

\citet{Shi10} finds that the SIR specification can explain the
clustering of agents' real estate investments. \citet{Fen04},
\citet{Hon04,Hon05}, \citet{Ivk07}, and \citet{Bro08} find similar
evidence suggesting that neighbors emulate each other's investment
choices.\footnote{\citet{Kin00}, p. 15, notes that ``there is
nothing so disturbing to one's well-being and judgement as to see
a friend get rich."} \citet{Ash55} shows that a considerable
fraction of people follow other's choices even when they know that
these choices are wrong.

\citet{Shi17} suggests that the SIR specification helps to study
the spread of economic ``narratives," such as the idea of the
``Laffer curve." In our model, such narratives come in the form of
the ``BRICS" countries, asian tigers, ``.com" stocks, securitized
mortgage debt, or ``crypto currencies." Related,
\citet{Dal64,Dal65} have sparked a literature, e.g.,
\citet{More04}, that uses epidemiological models to study the
spread rumors.


\citet{Sch90}, \citet{Bik92} and \citet{Ban92} provide Bayesian
frameworks of herd behavior, which can be used to generate
boom-bust sequences in asset prices within a traditional rational
expectation framework. In this interpretation, the current SIR
specification studies how certain pieces of Bayesian information
permeate a given population of investors. Regarding rationality,
\citet{Kin00}, p. 15-16, points out that speculative manias are
concentrated in assets which are hard to understand, such as
options, futures, tulip seeds during winter, real estate and land,
goods manufactured for export markets, and more generally foreign
exchange. Recent episodes like the ``.com" crash, the US housing
crisis with its complex mortgage backed securities, and the rise
and fall in crypto-currencies clearly fit this description.

Our model coexists with a large number of alternative theories of
boom and bust. First, any Walrasian model, which features multiple
equilibria, can explain abrupt changes in price simply as a move
from one equilibrium to another. Equilibrium models of portfolio
insurance, flights to quality, rational herding, informational
avalanches, and game theoretic formulations, such as models of
exchange rate crises, can generate drastic price changes within a
rational expectations (RE) frame. One key difference between
theses RE models and the present one is the speed of adjustment.
Rational agents ensure that markets reprice, more or less,
immediately. That is, moving from one equilibrium to another
produces an instantaneous change in price. The current model
generates smooth booms and smooth busts, which are in line with
the observation that many boom bust phenomena unfold over time
frames ranging from several month to decades.\footnote{Japanese
real estate saw steep appreciation for 30 years before it reached
its 1989-1991 peak. In turn, prices declined for roughly 25
years.}

Section \label{S1} introduces the model. Section \ref{S2}
considers the impact of rational expectations. Section \ref{S4}
compares our theoretical results to a number of historic boom-bust
episodes. Section \ref{S5} concludes.

\section{The Model}\label{S1}

We recall the SIR model,\footnote{See e.g., \citet{Hir03}, pp.
235-239.} and add the asset market later. The SIR model studies
how a disease spreads among a mass $N$ of agents. These agents
belong to three groups. First, there is a group of susceptible
agents $S$. Second, there is a group of infected agents $I$.
Finally, there are recovered agents $R$. When infected agents meet
susceptible agents, they transmit the virus at rate $\beta$.
Infected agents recover at rate $\gamma$. Finally, we work with
continuous time $t$, such that the epidemic is characterized by a
first-order differential equation:
\begin{eqnarray}
\dot{S}&=&-\beta IS,\label{e1}\\
\dot{I}&=&\beta IS-\gamma I,\label{e2}\\
\dot{R}&=&\gamma I,\label{e3}\\
N&=&S+I+R, \quad N=N_1+N_2+N_3.\label{e4}
\end{eqnarray}
Equation (\ref{e1}) describes how susceptible agents get infected.
Equation (\ref{e2}) tracks the mass of infected agents: it adds
newly infected and subtracts recovered agents. Equation (\ref{e3})
accounts for cured agents. Finally, (\ref{e4}) ensures that the
overall population remains constant over time; Starting in $t=0$
we have $N_1$ susceptible agents, $N_2$ infected agents, $N_3$
recovered agents, and a total population $N$.

Combining (\ref{e1})-(\ref{e4}) we note:

\begin{lemma}\label{lemma1} Given initial conditions $\{N_1,N_2,N_3\}, N_i>0, i=1,2,3$, the SIR model has one stable
steady state, $\{N-R_{\infty},0,R_{\infty}\}$, where
$R_{\infty}=-\frac{\gamma}{\beta}ln(N-R_{\infty})+C_R,$ and
$C_R=N_3+\frac{\gamma}{\beta}\ln(N_1)$.\end{lemma}
\begin{proof}  Combining  (\ref{e1}) with (\ref{e2}), we can integrate
\begin{eqnarray} I&=&-S+\frac{\gamma}{\beta}ln(S)+C_I, \quad C_I=N_1+N_2-\frac{\gamma}{\beta}\ln(N_1).\label{el1}\end{eqnarray}
Likewise (\ref{e1}) and (\ref{e3}) yield
\begin{eqnarray}R&=&-\frac{\gamma}{\beta}ln(S)+C_R,\quad
C_R=N_3+\frac{\gamma}{\beta}\ln(N_1).  \label{l2}\end{eqnarray}

From (\ref{e3}), we see that $I=0$ in steady state. Combining
$I=0$ with (\ref{el1}), (\ref{l2}) and (\ref{e4}), we have
$S_{\infty}=N-R_{\infty}$ and
$R_{\infty}=-\frac{\gamma}{\beta}ln(N-R_{\infty})+C_R$, where the
last equation defines a unique $R_{\infty}$.
\end{proof}

The steady state of Lemma \ref{lemma1} is reached via a transition
path, along which susceptible agents become infected, and infected
agents get cured. This transition path exhibits a peak in infected
agents:

\begin{lemma} If $N_1>\frac{\beta}{\gamma}$, then there exists a point in time $t_I^*>0,$ where the mass of infected agents $I$ peaks.\label{lemma2}
\end{lemma}
\begin{proof} To prove that there exists a unique peak in infected agents in a period $t_I^*$, we recall (\ref{el1}):
\begin{eqnarray}I=-S+\frac{\gamma}{\beta}ln(S)+C_I, \quad C_I=N_1+N_2-\frac{\gamma}{\beta}\ln(N_1).\end{eqnarray}
Taking derivatives yields
\begin{eqnarray}\frac{dI}{dS}=-1+\frac{\gamma}{\beta}\frac{1}{S}=0,\end{eqnarray}
and
\begin{eqnarray}\frac{d^2I}{dS^2}=-\frac{\gamma}{\beta}\frac{1}{S^2}<0.\end{eqnarray}
Hence, there exists a global infection maximum at
$S^*=\frac{\gamma}{\beta}$, where
$I^*=-\frac{\gamma}{\beta}+\frac{\gamma}{\beta}ln(\frac{\gamma}{\beta})-C_S$.
To relate our peak in infected agents to a period in time, we note
that equation (\ref{e1}) ensures that $\dot{S}<0$ at all points in
time, such that the condition $S(t)=S^*=\frac{\gamma}{\beta}$ can
be inverted $t^*_{I}=S^{-1}(\frac{\gamma}{\beta})$. That is, the
mass of infected agents peaks in time $t^*_{I}$, and at this point
in time, we have $dI/dt=0$ and $d^2I/dt^2<0$.
\end{proof}

To prove Lemma \ref{lemma2}, it was convenient to express the mass
of infected agents as a function $I(S)$. To study how this peak in
infected agents relates to stock prices, it will be necessary to
account for the fact that different cohorts of infected agents buy
and sell at different prices. To do so, we denote cohorts $c_v$ of
agents by the (vintage) period $t=v$, where they were first
infected:
\begin{eqnarray}c_v=\beta I_vS_v.\end{eqnarray}
Taking into account that infected agents get cured at rate
$\gamma$, and summing over all infected cohorts, yields the mass
of infected agents in period $t$:
\begin{eqnarray}I=\beta\int_0^{t}I_vS_ve^{-\gamma(t-v)}dv.\label{I0}\end{eqnarray}
Using (\ref{I0}), the condition for $t^*_I$ from Lemma
\ref{lemma2} can be written as:
\begin{eqnarray}\frac{dI}{dt}=\beta I(t)S(t)-\beta\gamma\int_0^{t}I_vS_ve^{-\gamma(t-v)}dv=0.\label{I1}\end{eqnarray}
Regarding (\ref{I1}) we note that inserting (\ref{I0}) brings us
back to (\ref{e2}). Expression (\ref{I1}) will, however, be useful
once we have to track dated asset purchases and sales of different
cohorts.

\subsection{Market}

We assume that susceptible agents hold one unit of currency each.
When agents get infected, they buy the asset. Once cured, agents
sell their position. Regarding the asset market, we assume that
there is an exogenously given (excess) supply function:
\begin{eqnarray}X=\phi(P), \quad \phi'()>0,\quad \quad \phi(P_0)=0, \label{e5}\end{eqnarray}
which interacts with the demand of infected agents. Equation
(\ref{e5}) may be interpreted as the supply of ``.com" stocks,
tulips, crypto currencies, flats, which come to the market,
respectively, via new IPOs, professional florists, bitcoin farms,
or the incumbent population in an in-fashion residential area.

To compute infected agents' demand, we recall
(\ref{e1})-(\ref{e2}), and note that a cohort $c_v=\beta I(v)S(v)$
of agents, who were infected in period $t=v$, depreciates
exponentially such that the cohort's period $t$ size is:
\begin{eqnarray}c_v(t)=\beta I(v)S(v)e^{-\gamma(t-v)}.\label{e6}\end{eqnarray}
Regarding the period $t$ asset holdings, of all infected cohorts
$v$, we have:\footnote{We suppress the holdings of agents $N_2$,
who are infected in $t=0$. Adding these holdings, inflates
notation and does not change the results. Alternatively, we may
assume that $N_2$ is small.}
\begin{eqnarray}X_{I}=\beta\int_0^{t}\frac{I_vS_v}{P(v)}e^{-\gamma(t-v)}dv.\label{e7}\end{eqnarray}
where $\frac{1}{P(v)}$ is the number of shares that infected
agents buy in period $t=v$, when they invest their one unit of
currency into the speculative asset. The variable $X_I$ denotes
the number of shares, held by all cohorts of infected agents, in
period $t$.

Combining asset supply (\ref{e5}) with demand (\ref{e7}), the
equilibrium price is:
\begin{eqnarray}\phi(P)=\beta\int_0^{t}\frac{I_vS_v}{P(v)}e^{-\gamma(t-v)}dv\quad \Leftrightarrow\quad
P=\phi^{-1}\Big(\beta\int_0^{t}\frac{I_vS_v}{P(v)}e^{-\gamma(t-v)}dv\Big).\label{e8}\end{eqnarray}
Evaluation of (\ref{e8}) yields:
\begin{prop} The asset's price peaks in period $t^*_p$, where $t^*_p<t^*_I$. The long-run, steady state, price is $P_0$. \label{p1}\end{prop}
\begin{proof} The asset price peaks when $\dot{P}=\frac{1}{\phi'(X_I)}\frac{dX_I}{dt}=0$. From (\ref{e7}) it
follows that this condition for $t^*_p$ can be written as
\begin{eqnarray}\frac{dX_I}{dt}=\beta
\frac{I(t)S(t)}{P(t)}-\beta\gamma\int_0^{t}\frac{I_vS_v}{P(v)}e^{-\gamma(t-v)}dv=0.\label{e9}\end{eqnarray}
Regarding the derivative $\frac{dX_I}{dt}$, we first note that
$\lim_{t\rightarrow 0}\frac{dX_I}{dt}>0$. That is, the price is
increasing at $t=0$. Moreover, we recall (\ref{e7}) and note that
$\lim_{t\rightarrow \infty}X_I=0$\footnote{To see this note that
$\beta\gamma\int_0^{t}\frac{I_vS_v}{P(v)}e^{-\gamma(t-v)}dv<\gamma\frac{I}{P_0}$
and recall that $\lim_{t\rightarrow\infty}I=0$.}, i.e., in the
long run steady state, the demand of infected speculators is zero.
Hence, the long run price is $P_0=\phi^{-1}(0)$.

\emph{Unique peak price:} We have seen that, starting with a price
$P_0$, the price is first increasing in time and then eventually
reverts to $P_0$ in the long-run. Hence, there must be at least
one $t^*_P$, where the price peaks. This period is implicitly
defined by condition (\ref{e9}). To show that the peak is unique,
we study the second-order condition:
\begin{eqnarray} \frac{d^2X_I}{dt^2}&&=-\frac{1}{P}\ddot{S}-\frac{1}{P^2}\dot{P}\dot{S}-\beta\gamma\Big[\beta
\frac{I(t)S(t)}{P(t)}-\beta\gamma\int_0^{t}\frac{I_vS_v}{P(v)}e^{-\gamma(t-v)}dv\Big]\\
&&=_{|\dot{P}=0,(\ref{e9})}-\frac{1}{P}\ddot{S}<0.\label{second}
\end{eqnarray}
To interpret (\ref{second}), we recall
$\ddot{S}_{|(\ref{e1})}=-\beta(\dot{I}S+I\dot{S})$ and rearrange
it using (\ref{e2}), such that:
\begin{eqnarray} \frac{\ddot{S}}{\dot{S}}\frac{1}{\beta}=-I+S-\frac{\gamma}{\beta}\label{e12}\end{eqnarray}
For the terms in (\ref{e12}), we note that (\ref{e1}) ensures
$\frac{1}{\dot{S}}<0$, and from Lemma \ref{lemma2} we know
$S-\frac{\gamma}{\beta}<0$ for all $t>t_I^*$. If the mass of
susceptible agents is sufficiently large,\footnote{As we have seen
in Lemma \ref{lemma2} we need $N_1>\frac{\gamma}{\beta}$ for a
peak in $I$ to exist. Moreover, we need only a small number of
initially infected agents to start an epidemic, i.e., we can pick
an arbitrarily small number $N_2$ for the number of initially
infected agents.} at the beginning of time, $\ddot{S}$ starts out
negative, but before time $t_I^*$ is reached, where
$S=\frac{\gamma}{\beta}$, it will turn positive, and stay positive
ever after. During the period where $\ddot{S}>0$, we may have
price peaks, but cannot have minima. Once, $\ddot{S}<0$ we may
have minima, but no maxima. That is, in $t=0$, prices start to
increase away from the starting level $P_0$. This price increase
first accelerates until the second derivative changes signs, and
the price increase decelerates until a peak is reached. After the
peak, prices start to decline.


Finally, it remains to show that $t^*_p<t^*_I$. To do so we
compare the condition (\ref{e9}), for the period where the price
peaks, with the condition for $t^*_I$, where the number of
infected agents peaks. The mass of infected agents peaks when:
\begin{eqnarray}\frac{dI}{dt}=\beta I(t)S(t)-\beta\gamma\int_0^{t}I_vS_ve^{-\gamma(t-v)}dv=0.\label{e111}\end{eqnarray}
Multiplying (\ref{e9}), the condition for the period where the
price peaks, with $P(t)$, yields
\begin{eqnarray}\frac{dX_I}{dt}P(t)=\beta
I(t)S(t)-\beta\gamma\int_0^{t}\frac{I_vS_vP(t)}{P(v)}e^{-\gamma(t-v)}dv=0.\label{e11}\end{eqnarray}
The difference between (\ref{e111}) and (\ref{e11}) is the term
$\frac{P(t)}{P(v)}$. Regarding this term, we note that, as long as
the price is increasing between time $0$ and $t$, we have
$\frac{P(t)}{P(v)}>1$ for all $v<t$. Accordingly, for any
arbitrary period of time $t$, for which the price is increasing,
we have $\frac{dI}{dt}>\frac{dX_I}{dt}P(t)$. Hence, when the price
peaks at $t^*_p$, we have $\frac{dX_I}{dt}P(t^*_p)=0$ and
$\frac{dI}{dt}>0$. That is, when the price peaks, the mass of
infected agents is still increasing in time.

Put differently, early buyers have amassed large wealth, and when
they start to sell, they have a greater weight in the market than
those agents who buy in late with their one unit of
currency.\end{proof}






\subsection{Depressions}

In the current model we have discussed how investors, who infect
each other with euphoria, can bid up prices. The same argument
works in reverse when depression spreads. That is, we may model an
infected agent as someone who adopts the view that stocks or real
estate prices will never appreciate again. In that case, the same
mechanism works in reverse, creating a U-shaped price pattern.
Regarding this depression, we can once again show that asset
prices bottom before the selling sentiment does. 

\section{Rational Expectations}\label{S2}

\emph{At the top of the market there is hesitation, as new
recruits to speculation are balanced by insiders who withdraw.}
\citet{Kin00} p. 17

So far we have assumed that agents simply close their position
once they recover from their false investment thesis. Put
differently, when agents understood that they owned an over valued
asset, they sold. Another way to model recovering agents is to
assume that they fully understand the model that they operate in.

When cured agents know the entire model, they can find themselves
in two different situations. First, the agent gets cured after the
price has peaked. In this scenario it is clearly rational to sell
immediately. Second, the agent gets cured while the market price
is still increasing. In this second scenario, cured agents will
keep the asset for as long as it appreciates. To do so, they have
to compute the point in time when the market peaks. To find this
point, they must compute the asset holdings of all other agents
who are also waiting to exit at the top. Second, they must know
the asset holdings of those who are currently infected, but will
get cured before the price peaks. Finally, the mass of susceptible
investors who will buy in the future is important, since these
investors provide the liquidity that rational sellers need to exit
their positions.



Two points in time are thus of particular importance. First, there
is the point in time $t_1,$ where the price peaks, and rational
agents start to sell to newly infected buyers. At $t_1$ the sum of
future net buying of newly infected agents, must be sufficient to
absorb all the holdings that rational agents, who were waiting for
the price to peak, have accumulated up until $t_1$. To calculate
the net buying of newly infected agents, we have to identify the
point in time $t_2$ where the buying of newly infected agents is
exactly offset by the selling of newly cured agents. Time $t_2$ is
the point where the last rational agent, who was waiting for the
market top, must have sold out. After $t_2$, buying of newly
infected agents falls short of the selling of newly cured agents
and prices decline.

We start with $t_2$, the point in time where the buying of newly
infected agents is exactly offset by the selling of newly cured
agents:
\begin{eqnarray}\frac{dX_I}{dt}=\beta
\frac{I(t)S(t)}{P(t)}-\beta\gamma\int_0^{t}\frac{I_vS_v}{P(v)}e^{-\gamma(t-v)}dv=0.\label{e99}\end{eqnarray}
This condition is the same as (\ref{e9}) except for the price
history, i.e., the values $P(v)$ are different since cured agents
did not sell immediately.\footnote{Period $t_2$ therefore does not
coincide with $t_P^*$.}

Period $t_1$, is the period when rational agents start to sell:
\begin{eqnarray}\int_{0}^{t_1}\Big(1-e^{-\gamma(t_1-v)}\Big)\beta
\frac{I_vS_v}{P(v)}dv=\beta\int_{t_1}^{t_2}\frac{I_vS_v}{P(v)}e^{-\gamma(t_2-v)}dv=0.\label{e999}\end{eqnarray}
The RHS of (\ref{e999}) describes the net buying of newly infected
agents between $t_1$ and $t_2$. This must be equal/sufficient to
absorb the assets that cured agents accumulated (LHS) between
$t=0$ and $t_1$ with the motive to sell at the market top. Put
differently, (\ref{e999}) ensures that the entire selling of cured
agents can be absorbed through the buying of newly infected
agents. And once period $t_2$ is reached, the holdings of cured
agents, who waited for the top, have been passed on completely to
newly infected agents. Finally, the price at which rational agents
sell is:
\begin{eqnarray}P^*=\phi^{-1}\Big(\int_{0}^{t_1}\beta\frac{I_vS_v}{P(v)}dt\Big).\label{e9999}\end{eqnarray}
Where (\ref{e9999}) reflects that there is no selling before time
$t_1$ is reached. Taking (\ref{e99})-(\ref{e9999}) together we
have:
\begin{prop}\label{p2} (Broad-Peak-Theorem): The market price remains constant at its peaks $P^*$ as long as $t\in[t_1,t_2]$.\end{prop}
Proposition \ref{p2} relies on the assumption that cured agents
know the correct model and all its coefficients. Moreover, they
are able to solve the model. Proposition \ref{p2} should thus be
interpreted as a reference point, which we can use as a comparison
for the ``fever peak" which we derived in Proposition \ref{p1}.
The fever peak relied on the assumption that recovered agents
cannot fully understand the model, and thus they cannot correctly
anticipate the top. Hence, they simply close their position. Put
differently, these agents just discard the wrong thesis they used
to buy the asset, but they do not replace this thesis with a new
thesis, which tells them, e.g., that the price may appreciate
further. To close our comparison of these two cases, we make three
remarks. Whenever necessary, we denote variables that correspond
to the case where cured agents simply sell by a subscript M. Cases
where cured agents form rational expectations are denoted by $RE$.
We begin by noting that the price increase is more rapid when
cured agents from rational expectations:
\begin{rem} For the period where $t\in(0,t_1]$ we have $P_M<P_{RE}$. \label{rem1}\end{rem}
\begin{proof} In both scenarios, the buying is
governed by the same infection process. The scenarios differ
regarding selling. In the RE setting, where cured agents form
rational expectations, there is no selling during the period
$t\in[0,t_1]$. Asset holdings, and prices, are thus always higher in the RE setting than in setting $M$ where agents sell immediately after being cured.
\end{proof}

\begin{rem} The peak price is lower under rational expectations $P_{RE}^*<P_M^*$.\label{rem3} \end{rem}
\begin{proof} Let us first recall that all rational agents have sold out, by the time that $t_2$ is
reached. Hence, the mass of infected agents who hold the
speculative asset is the same in $t_2$ across both scenarios RE
and M. To proof our remark by contradiction, we now assume that
$P_M<P_{RE}$ for all $t\in[0,t_2)$. Under this assumption period
$t_2$ asset holdings are:
\begin{eqnarray} X_{RE}=\beta\int_0^{t_2}\frac{I_vS_v}{P_{RE}(v)}e^{-\gamma(t-v)}dv<\beta\int_0^{t_2}\frac{I_vS_v}{P_{M}(v)}e^{-\gamma(t-v)}dv=X_M. \label{r3}\end{eqnarray}
This however means that $P(X_{RE})<P(X_M)$, and contradicts our
assumption $P_M<P_{RE}$ for all $t\in[0,t_2)$. Hence, the price
$P_M$ must exceed $P_{RE}^*$ before $t_2$ is reached.
\end{proof}

\begin{rem} $t_1<t_P^*<t_2<t_I^*$.\end{rem}
\begin{proof} $t_1<t_P^*$: The price $P_{RE}$ peaks in period $t_1$ and remains
at this level until $t_2$ is reached. To see that $t_1<t_P^*$, we
recall Remark \ref{rem1}, i.e., that $P_{RE}>P_M$ for all
$t\in(0,t_1]$. At the same time, Remark \ref{rem3} shows that
$P_M$ must exceed $P_{RE}$ before $t_2$ is reached. Hence, the
date $t_P^*$, where $P_M$ peaks, must be such that $t_1<t_P^*$.


$t_2<t_I^*$ follows immediately from the same argument given in
the proof of Proposition \ref{p2}, which relies on the comparison
of equations (\ref{e111}) and (\ref{e11}). That is, in period
$t_2$, all newly cured agents bought at prices less or equal
$P^*_{RE}$. Thus, given that their one Dollar investment
appreciated, they have a higher monetary weight than the newly
infected buyers. Hence, at $t_2$, the mass of infected agents must
still be increasing to absorb the selling. An alternative way, to
observe the same thing, is to show that the price is already in
decline at $t_I^*$.



Inequality $t_P^*<t_2$: We have shown in Remark \ref{rem3} that
$P_M$ must exceed $P_{RE}^*$ before $t_2$. Does the peak in $P_M$
also occur before $t_2$? If $P_{M}>P_{RE}$ then $X_M>X_{RE}$.
taking derivatives, at $t_2$, yields:
\begin{eqnarray} 0=\frac{IS}{P_{RE}}-\gamma\int_0^{t_2}\frac{I_vS_v}{P_{RE}(v)}e^{-\gamma(t-v)}dv\gtreqqless
\frac{IS}{P_{M}}-\gamma\int_0^{t_2}\frac{I_vS_v}{P_{M}(v)}e^{-\gamma(t-v)}dv.
\label{r4}\end{eqnarray} For $t_2<t^*_P$, we must have
$P_M>P_{RE}$ in $t_2$ This means that
$\frac{IS}{P_{M}}<\frac{IS}{P_{RE}}$ since
$P=P(\gamma\int_0^{t_2}\frac{I_vS_v}{P_{RE}(v)}e^{-\gamma(t-v)}dv)$,
we also know that
$\gamma\int_0^{t_2}\frac{I_vS_v}{P_{RE}(v)}e^{-\gamma(t-v)}dv<\gamma\int_0^{t_2}\frac{I_vS_v}{P_{M}(v)}e^{-\gamma(t-v)}dv$.
Rearranging (\ref{r4}) thus gives:
\begin{eqnarray} IS\Big(\frac{1}{P_{RE}}-\frac{1}{P_{M}}\Big)-
\Big(\gamma\int_0^{t_2}\frac{I_vS_v}{P_{RE}(v)}e^{-\gamma(t-v)}dv-\gamma\int_0^{t_2}\frac{I_vS_v}{P_{M}(v)}e^{-\gamma(t-v)}dv\Big)>0
\end{eqnarray}
Hence, $\frac{dX_M}{dt}(t_2)<0$, i.e., the price $P_M$ is in
decline at $t_2$, and its peak must have occurred earlier. That
is, $t_P^*<t_2$.
\end{proof}






\section{Examples}\label{S4}

In this section, we use data from Google Ngrams and Goolde trends
to study some of our predictions. Across cases, there is strong
support for the epidemiological fever-peak. The search
queries/literature mentions for particular assets show a
pronounced hump-shape pattern, which coincides with a hump shaped
pattern in these assets' prices. Moreover, in line with our
model's prediction, prices peak earlier than search
queries/literature mentions. 

\emph{Bitcoin:} Using Google trends, we have the search history
for the term ``Bitcoin." Bitcoin search peaked in the period
between December 17th and 23rd at 100. Before, Bitcoin search was
at a level of 3 between 2013 and April 2017. Dollar prices for
Bitcoin went from 1200 in April 2017 to their 19.587 peak on 16th
of December 2018. While search for Bitcoin continued to rise
sharply between December 17th and 23rd, prices dropped to around
14000\footnote{Prices are taken from Yahoo Finance. December 23rd
highs and lows were roughly 15000 and 13000.} on December 23, and
have continued to decline alongside with search queries, to 3800
respectively 10 in December 2018.


\emph{The ``.com" boom:} Using Google Ngrams data for the period
1980-2008, we find that the frequency with which the terms
``Nasdaq," ``MSFT" and ``AMZN" are mentioned 
follow the hump-shaped pattern suggested by the SIR model. This
pattern extends to the corresponding prices. Regarding the time
line of peaks, we note that mentions of the ``Nasdaq" peaked in
2001, while the Nasdaq composite index peaked in March 2000. The
MSFT (Microsoft) stock peaked in late 1999, mentions of the
``MSFT" ticker symbol peaked in 2004. Mentions of the AMZN
(Amazon) ticker peaked in 2002 while the stock peaked in late
1999. Mentions of ``new economy" and ``software" both peaked in
2002.


\emph{Japanese Nikkei index:} Using Google Ngrams (case
insensitive), we view the term ``Made in Japan" as a proxy for the
sentiment towards the Japanese Nikkei industrial index. The Nikkei
peak in 1989 leads the ``Made in Japan" peak between
1990-1994.




\section{Conclusion}\label{S5}

We model booms and busts in asset prices as an epidemic process:
agents infect each other with their investment ideas, valuation
techniques, or believes in future technological developments. To
model how such ideas permeate a large population of susceptible
investors we have used the classic SIR specification from the
literature on epidemiology.

Our model generates sentiment driven booms-bust cycles, which
market practitioners emphasize. In addition to producing such
price swings, our model can be used to perform comparative
statics. In particular, we find that rational agents (i)
accelerate booms (ii) lower peak prices and (iii) make for broad,
drawn-out, market tops. Finally, regardless of whether cured
agents are rational or not, prices peak before sentiment, i.e. the
mass of infected agents, does.


Google data indicate that interest in speculative assets indeed
exhibits boom-bust patterns akin to the SIR model's infection
peak. Moreover, the peak in price tends to precede the peak in
queries. The ``Bitcoin" boom with its sharp peak was more in line
with our scenario, where cured agents are best advised to sell
quickly. Deeper markets had broader peaks, which give well
informed investors more time to trade with late stage optimists.

\newpage

\begin{appendix}

\addcontentsline{toc}{section}{References}
\markboth{References}{References}
\bibliographystyle{apalike}
\bibliography{References}

\end{appendix}


\end{document}